\documentclass{article}

  \usepackage{fullpage}
  \usepackage{tikz}
  \usepackage{xspace}
  \usepackage{makecell}
  \usepackage{amsmath,amsfonts,amsthm,amssymb}
  \usepackage{environ}
  \usepackage[ruled,noline,noend]{algorithm2e}
  \SetKw{Continue}{continue}
  \usepackage{tabularx}
  \usepackage{comment}
  \usepackage{todonotes}
  \usepackage{enumerate}
  \usepackage{mathtools}
  \usepackage{setspace}
  \usepackage[hidelinks]{hyperref}
  \usepackage[capitalise]{cleveref}
  \usepackage{authblk}
  \usetikzlibrary{snakes, calc, patterns, arrows}
  \usetikzlibrary{decorations.pathreplacing}
  \usepackage{tikz}
  \usepackage{amssymb}
  
  \definecolor{darkblue}{rgb}{0.2,0.2,0.6}
  \definecolor{darkteal}{rgb}{0.14,0.36,0.38}
  \definecolor{darkbrown}{rgb}{0.39,0.26,0.13}
  
  \newtheorem{theorem}{Theorem}[section]
  
  \newtheorem{lemma}[theorem]{Lemma}
  \newtheorem{observation}[theorem]{Observation}
  \newtheorem{proposition}[theorem]{Proposition}
  \newtheorem{corollary}[theorem]{Corollary}
  \newtheorem{claim}[theorem]{Claim}
  
  \theoremstyle{remark}

  \SetCommentSty{textnormal}
  
\newcommand{\defproblem}[3]{
  \vspace{2mm}
  \begin{center}
  \noindent\fbox{
  \begin{minipage}{0.9\textwidth}
  \textsc{#1}

  \smallskip
  \noindent
  {\bf{Input:}} #2
  
  \smallskip
  \noindent
  {\bf{Output:}} #3
  \end{minipage}
  }
  \end{center}
  \vspace{2mm}
}

\newcommand{\defproblemq}[3]{
  \vspace{2mm}
  \begin{center}
  \noindent\fbox{
  \begin{minipage}{0.8\textwidth}
  \textsc{#1}

  \smallskip
  \noindent
  {\bf{Input:}} #2
  
  \smallskip
  \noindent
  {\bf{Query:}} #3
  \end{minipage}
  }
  \end{center}
  \vspace{2mm}
}

  \newcommand{\ceil}[1]{\left\lceil #1 \right\rceil}
  \newcommand{\Oh}{\mathcal{O}}
  \newcommand{\cO}{\mathcal{O}}
  
  \newcommand{\dd}{\mathinner{.\,.}}

  \newcommand{\Borders}{\mathit{Borders}}
  \newcommand{\lcp}{\mathsf{lcp}}
  \newcommand{\lcs}{\mathsf{lcs}}
  
  \newcommand{\SeedSets}{\mathit{SeedSets}}

  \newcommand{\LCovP}{\textsc{CoveredPref}}
  \newcommand{\SeededPref}{\textsc{SeededBasicPref}}
  \newcommand{\IsCover}{\textsc{IsCover}}

  \newcommand{\Cov}{\mathit{Cov}}
  \newcommand{\C}{\allcov}

  \newcommand{\per}{\mathsf{per}}
  \newcommand{\expon}{\mathsf{exp}}
  \newcommand{\run}{\mathsf{run}}
  \newcommand{\result}{\mathsf{result}}
  
  \newcommand{\shcov}{\textsc{MinCover}}
  \newcommand{\allcov}{\textsc{AllCovers}}
  
  \newcommand{\T}{\mathcal{T}}

\begin{document}

\title{Internal Quasiperiod Queries}

\author[1]{Maxime Crochemore}
\author[1]{Costas S.\ Iliopoulos}
\author[2]{Jakub Radoszewski}
\author[2]{Wojciech~Rytter}
\author[2]{Juliusz Straszy{\'n}ski}
\author[2]{Tomasz Wale{\'n}}
\author[2]{Wiktor Zuba}

\affil[1]{
Department of Informatics, King's College London, UK\\\texttt{[maxime.crochemore,c.iliopoulos]@kcl.ac.uk}
}
\affil[2]{
Institute of Informatics, University of Warsaw, Poland\\
\texttt{[jrad,rytter,jks,walen,w.zuba]@mimuw.edu.pl}
}

\maketitle

\setcounter{footnote}{0} 

\begin{abstract}
Internal pattern matching requires one to answer queries about factors of a given string. Many results are known on answering internal period queries, asking for the periods of a given factor. In this paper we investigate (for the first time) internal queries asking for covers (also known as quasiperiods) of a given factor.  We propose a data structure that answers such queries in $\Oh(\log n \log \log n)$ time for the shortest cover and in $\Oh(\log n (\log \log n)^2)$ time for a representation of all the covers, after $\Oh(n \log n)$ time and space preprocessing.
\end{abstract}

\section{Introduction}
A \emph{cover} (also known as a quasiperiod) is a weak version of a period. It is a factor of a text $T$ whose occurrences cover all positions in $T$; see \cref{fig:ex_cover}.
The notion of cover is well-studied in the off-line model.
Linear-time algorithms for computing the shortest cover and all the covers of a string of length $n$ were proposed in \cite{DBLP:journals/ipl/ApostolicoFI91} and \cite{DBLP:journals/ipl/MooreS94,DBLP:journals/ipl/MooreS95}, respectively.
Moreover, linear-time algorithms for computing shortest and longest covers of all prefixes of a string are known; see \cite{DBLP:journals/ipl/Breslauer92} and \cite{DBLP:journals/algorithmica/LiS02}, respectively.
Covers were also studied in parallel~\cite{DBLP:journals/iandc/BerkmanIP95,DBLP:journals/ipl/Breslauer94} and streaming~\cite{DBLP:conf/cpm/GawrychowskiRS19} models of computation. Definitions of other variants of quasiperiodicity can be found in the survey~\cite{Czajka}. In this work we introduce covers to the internal pattern matching model~\cite{DBLP:conf/soda/KociumakaRRW15}.

\begin{figure}[htpb]
   \begin{center}
    \begin{tikzpicture}[scale=0.53]
\draw (-1,0) node[above] {$T:$};
\foreach \x in {1,3,4,6,8,9,11,13} {
\draw[thick,darkblue] (\x/2,0) node [above] {{a}};
}
\foreach \x in {2,5,7,10,12} {
\draw[thick,darkblue] (\x/2,0) node [above] {{b}};
}

\foreach \x in {1,4,9} {
   \draw[thick,olive,yshift=0.2cm] (\x/2-0.2,0.5) -- (\x/2-0.2,0.6) -- (\x/2+1.2,0.6) -- (\x/2+1.2,0.5);
}
\foreach \x in {6,11} {
   \draw[thick,olive,yshift=0.5cm] (\x/2-0.2,0.5) -- (\x/2-0.2,0.6) -- (\x/2+1.2,0.6) -- (\x/2+1.2,0.5);
}
\foreach \x in {2} {
   \draw[thick,cyan] (\x/2-0.2,0) -- (\x/2-0.2,-0.1) -- (\x/2+3.2,-0.1) -- (\x/2+3.2,0);
}
\foreach \x in {7} {
   \draw[thick,cyan,yshift=-0.3cm] (\x/2-0.2,0) -- (\x/2-0.2,-0.1) -- (\x/2+3.2,-0.1) -- (\x/2+3.2,0);
}
\end{tikzpicture}   
  \end{center}
    \caption{
  $\shcov(T)=aba$ is the shortest cover of $T$ and $\shcov(T[2 \dd 13])=baababa$ is the shortest cover of its  suffix
  of length 12.
    }\label{fig:ex_cover}
  \end{figure}

In the internal pattern matching model, a text $T$ of length $n$ is given in advance and the goal is to answer queries related to factors of the text.
One of the basic internal queries in texts are \emph{period queries}, that were introduced in~\cite{DBLP:conf/spire/KociumakaRRW12} (actually, internal primitivity queries were considered even earlier~\cite{DBLP:conf/spire/CrochemoreIKRRW10,DBLP:journals/tcs/CrochemoreIKRRW14}).
A period query requires one to compute all the periods of a given factor of $T$.
It is known that they can be expressed as $\Oh(\log n)$ arithmetic sequences.
The fastest known algorithm answering period queries is from~\cite{DBLP:conf/soda/KociumakaRRW15}. It uses a data structure of $\Oh(n)$ size that can be constructed in $\Oh(n)$ expected time and answers period queries in $\Oh(\log n)$ time (a deterministic construction of this data structure was given in~\cite{tomeksthesis}).
A special case of period queries are \emph{two-period queries}, which ask for the shortest period of a factor that is known to be periodic. In~\cite{DBLP:conf/soda/KociumakaRRW15} it was shown that two-period queries can be answered in constant time after $\Oh(n)$-time preprocessing. Another algorithm for answering such queries was proposed in~\cite{DBLP:journals/siamcomp/BannaiIINTT17}.

Let us denote by $\shcov(S)$ and $\allcov(S)$, respectively, the length of the shortest cover and the lengths of all covers of a string $S$. Similarly as in the case of periods, it can be shown that the set $\allcov(S)$ can be expressed as a union of $\Oh(\log |S|)$ pairwise disjoint arithmetic sequences. We consider data structures that allow to efficiently answer these queries in the internal model.

\defproblemq{Internal quasiperiod queries}
{A text $T$ of length $n$}{For any factor $S$ of $T$, compute  $\shcov(S)$ or\\
\hspace*{0.4cm} $\allcov(S)$ after efficient preprocessing of the text $T$}

Recently~\cite{DBLP:conf/walcom/CrochemoreIRRSW20} we have shown how to compute the shortest cover of each cyclic shift of a string $T$ of length $n$, that is, the shortest cover of each length-$|T|$ factor of $T^2$, in $\Oh(n \log n)$ total time. This work can be viewed as a generalization of~\cite{DBLP:conf/walcom/CrochemoreIRRSW20} to computing covers of any factor of a string.
It also generalizes the earlier works on computing covers of prefixes of a string~\cite{DBLP:journals/ipl/Breslauer92,DBLP:journals/algorithmica/LiS02}.

\paragraph{\bf Our results.}
We show that $\shcov$ and $\allcov$ queries can be answered in $\Oh(\log n \,\log \log n)$ time and $\Oh(\log n \,(\log \log n)^2)$ time, respectively, with a data structure that uses $\Oh(n \log n)$ space and can be constructed in $\Oh(n \log n)$ time. In particular, the time required to answer an $\allcov$ query is slower by only a $\mathrm{poly\,log\,log\,}n$ factor from optimal. Moreover, we show that any $m$ $\shcov$ or $\allcov$ queries can be answered off-line in $\Oh((n+m) \log n)$ and $\Oh((n+m) \log n \log \log n)$ time, respectively, and $\Oh(n+m)$ space. In particular, the former matches the complexity of the best known solution for computing shortest covers of all cyclic shifts of a string~\cite{DBLP:conf/walcom/CrochemoreIRRSW20}, despite being far more general. We assume the word RAM model of computation with word size $\Omega(\log n)$.

\paragraph{\bf Our approach.} Our main tool are \emph{seeds}, a known generalization of the notion of cover. A seed is defined as a cover of a superstring of the text~\cite{seedsnlogn}. A representation of all seeds of a string $T$, denoted here $\mathit{SeedSet}(T)$, can be computed in linear time~\cite{10.1145/3386369}.
We will frequently extract individual seeds from $\mathit{SeedSet}(T)$; each time such an auxiliary query needs $\Oh(\log \log n)$ time. Consequently, $\log \log n$ is a frequent
factor in our query times related to internal covers.

We construct a tree-structure (static range tree) of so-called {\it basic factors} of a string. For each basic factor $F$ we store a compact representation of the set $\mathit{SeedSet}(F)$. 
The crucial point is that the total length of all these factors is $\Oh(n \log n)$ and every other factor can be represented,
using the tree-structure, as a
concatenation of $\Oh(\log n)$ basic factors. Representations of seed-sets of basic factors are precomputed. 
Then, upon an internal query related to a specific factor $S$, we
decompose $S$ into concatenation of  basic factors $F_1,F_2,\dots,F_k$. 
Intuitively, the representation of the set of covers or (in easier queries) the shortest cover will be computed as a
``composition''
of $\mathit{SeedSet}(F_1),\mathit{SeedSet}(F_2),\ldots, \mathit{SeedSet}(F_k)$, followed by adjusting it to border conditions using internal pattern matching. 
To get efficiency, when quering about covers of a factor $S$,
we do not compute the whole representation of $\mathit{SeedSet}(S)$ (these representations are only precomputed for basic factors). 

Finally, several stringology tools related to properties of covers and string periodicity are used to improve $\mathtt{polylog}\,n$-factors in the query time that would result from a direct application this approach.

\section{Preliminaries}
We consider a text $T$ of length $n$ over an integer alphabet $\{0,\ldots,n^{\Oh(1)}\}$. If this is not the case, its letters can be sorted and renumbered in $\Oh(n \log n)$ time, which does not influence the preprocessing time of our data structure.

For a string $S$, by $|S|$ we denote its length and by $S[i]$ we denote its $i$th letter ($i=1,\ldots,|S|$).
By $S[i \dd j]$ we denote the string $S[i] \dots S[j]$ called a factor of $S$; it is a prefix if $i=1$ and a suffix if $j=|S|$.
A factor that occurs both as a prefix and as a suffix of $S$ is called a border of $S$. A factor is proper if it is shorter than the string itself.
A positive integer $p$ is called a period of $S$ if $S[i]=S[i+p]$ holds for all $i=1,\ldots,|S|-p$.
By $\per(S)$ we denote the smallest period of $S$.
A string $S$ is called periodic if $|S| \ge 2\per(S)$ and aperiodic otherwise.
If $S=XY$, then any string of the form $YX$ is called a cyclic shift of $S$.
We use the following simple fact related to covers.
\begin{observation}\label{obsbasic}
Let $A$, $B$, $C$ be strings such that $|A| < |B| < |C|$.
\begin{enumerate}[(a)]
\item If $A$ is a cover of $B$ and $B$ is a cover of $C$, then $A$ is a cover of $C$.
\item If $B$ is a border of $C$ and $A$ is a cover of $C$, then $A$ is a cover of $B$.
\end{enumerate}
\end{observation}

Below we list several algorithmic tools used later in the paper.

\subsection{Queries Related to Suffix Trees and Arrays}
A range minimum query on array $A[1 \dd n]$ requires to compute $\min\{A[i],\ldots,A[j]\}$.

\begin{lemma}[\cite{DBLP:conf/latin/BenderF00}]\label{RMQ}
Range minimum queries on an array of size $n$ can be answered in $\Oh(1)$ time after $\Oh(n)$-time preprocessing.
\end{lemma}

By $\lcp(i,j)$ ($\lcs(i,j)$) we denote the length of the longest common prefix of $T[i \dd n]$ and $T[j \dd n]$ (longest common suffix of $T[1 \dd i]$ and $T[1 \dd j]$, respectively). Such queries are called longest common extension (LCE) queries. The following lemma is obtained by using range minimum queries on suffix arrays.

\begin{lemma}[\cite{DBLP:conf/latin/BenderF00,DBLP:journals/jacm/KarkkainenSB06}]\label{LCE}
After $\Oh(n)$-time preprocessing, one can answer LCE queries for $T$ in $\Oh(1)$ time.
\end{lemma}

The suffix tree of $T$, denoted as $\T(T)$, is a compact trie of all suffixes of $T$. Each implicit or explicit node of $\T(T)$ corresponds to a factor of $T$, called its \emph{string label}.
The \emph{string depth} of a node of $\T(T)$ is the length of its string label.

We use \emph{weighted ancestor (WA) queries} on a suffix tree. Such queries, given an explicit node $v$ and an integer value $\ell$ that does not exceed the string depth of $v$, ask for the highest explicit ancestor $u$ of $v$ with string depth at least $\ell$.

\begin{lemma}[\cite{DBLP:journals/talg/AmirLLS07,10.1145/3386369}]\label{WAQ}
Let $\T(T)$ be the suffix tree of $T$. WA queries on $\T(T)$ can be answered in $\gamma_n = \Oh(\log \log n)$ 
time after $\Oh(n)$-time preprocessing. Moreover, any $m$ WA queries on $\T(T)$ can be answered off-line in $\Oh(n+m)$ time.
\end{lemma}

\subsection{Internal Pattern Matching (IPM)}
The data structure for IPM queries is built upon a text $T$ and allows efficient
location of all occurrences of one factor $X$ of $T$ inside another factor $Y$ of $T$, where $|Y| \le 2|X|$.

\begin{lemma}[\cite{DBLP:conf/soda/KociumakaRRW15}]\label{IPM} The result of an IPM query is a single arithmetic sequence. After linear-time preprocessing one can answer IPM queries for $T$ in $\Oh(1)$ time. 
\end{lemma}

A period query, for a given factor $X$ of text $T$, returns a compact representation of all the periods of $X$ (as a set of $\Oh(\log n)$ arithmetic sequences). 

\begin{lemma}[\cite{DBLP:conf/soda/KociumakaRRW15}]\label{PQ}
After $\Oh(n)$ time and space preprocessing, for any factor of $T$ we can answer a period query in $\Oh(\log n)$ time. 
\end{lemma}

The data structures of \cref{IPM,PQ} are constructed in $\Oh(n)$ expected time. These constructions were made worst-case in \cite{tomeksthesis}.

\subsection{Static Range Trees}
A \emph{basic interval} is an interval $[a\dd a+2^i)$ such that $2^i$ divides $a-1$. We assume w.l.o.g.\ that $n$ is a power of two. We consider a static range tree structure whose nodes correspond to basic subintervals of $[1\dd n]$ and a non-leaf node has children corresponding to the two halves of the interval. (See e.g.\ \cite{DBLP:conf/lata/KociumakaRRSWZ19}). The total number of basic intervals is $\Oh(n)$. Using the tree, every interval $[i\dd j]$ can be decomposed into $\Oh(\log n)$ pairwise disjoint basic intervals. The decomposition can be computed in $\Oh(\log n)$ time by inspecting the paths from the leaves corresponding to $i$ and $j$ to their lowest common ancestor.
A \emph{basic factor} of $T$ is a factor that corresponds to positions from a basic interval. 

\subsection{Seeds}
We say that a string $S$ is a \emph{seed} of a string $U$ if $S$ is a factor of $U$ and $S$ is a cover of a string $U'$ such that $U$ is a factor of $U'$; see \cref{fig2}.
The second point of the lemma below follows from \cref{WAQ}.

\begin{figure}[htpb]
     \begin{center}
\begin{tikzpicture}[scale=0.53]
\clip (-0.125,-0.41) rectangle (8.225,0.96);
\foreach \x in {0,1,3,4,6,8,9,11,13,14,16} {
\draw (\x/2,0) node [above] {{a}};
}
\foreach \x in {2,5,7,10,12,15} {
\draw (\x/2,0) node [above] {{b}};
}
\
\foreach \x in {-2,4,9,14} {
   \draw[brown!60!black,thick,yshift=0.25cm] (\x/2-0.2,0.5) -- (\x/2-0.2,0.6) -- (\x/2+1.2,0.6) -- (\x/2+1.2,0.5);
}
\foreach \x in {1,6,11,16} {
   \draw[olive,thick] (\x/2-0.2,0) -- (\x/2-0.2,-0.1) -- (\x/2+1.2,-0.1) -- (\x/2+1.2,0);
}
\end{tikzpicture}
  \end{center}
   \begin{center}
    \begin{tikzpicture}[scale=0.53]
\clip (-0.125,-0.41) rectangle (8.225,0.96);
\foreach \x in {0,1,3,4,6,8,9,11,13,14,16} {
\draw (\x/2,0) node [above] {{a}};
}
\foreach \x in {2,5,7,10,12,15} {
\draw (\x/2,0) node [above] {{b}};
}
\
\foreach \x in {-2,6,14} {
   \draw[olive,thick,yshift=0.15cm] (\x/2-0.2,0.5) -- (\x/2-0.2,0.6) -- (\x/2+2.2,0.6) -- (\x/2+2.2,0.5);
}
\foreach \x in {1,11} {
   \draw[brown!60!black,thick] (\x/2-0.2,0) -- (\x/2-0.2,-0.1) -- (\x/2+2.2,-0.1) -- (\x/2+2.2,0);
}
\end{tikzpicture}
  \end{center}

    \caption{\label{fig:ex_seed}
       The strings $aba,\, abaab$ are seeds of the given string (as well as strings $abaaba$, $abaababa$, $abaababaa$). 
    }\label{fig2}
  \end{figure}
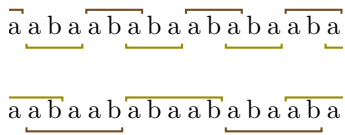

\begin{lemma}[\cite{10.1145/3386369}]\label{lem:seeds} 
\begin{enumerate}[(a)]
\item\label{xa} All the seeds of $T$ can be represented as a collection of a linear number of disjoint paths in the suffix tree $\T(T)$. Moreover, this representation can be computed in $\Oh(n)$ time if $T$ is over an integer alphabet.
\item\label{xb} After $\Oh(n)$ time preprocessing we can check if a given factor of $T$ is a seed of $T$ in $\Oh(\gamma_n)$ time.
\end{enumerate}
\end{lemma}

Our main data structure is a static range tree $\SeedSets(T)$ which stores all seeds of every basic factor of $T$ represented as a collection of paths in its suffix tree. Actually, only seeds of length at most half of a string will be of interest; see \cref{seedsets}. 

\begin{figure}
    \centering
\begin{tikzpicture}[scale=1.05,blue!20!black,-, >=stealth, auto, semithick,
    level 1/.style={sibling distance=110mm},
    level 2/.style={sibling distance=55mm},
    level 3/.style={sibling distance=25mm},
    level 4/.style={sibling distance=15mm},
    tree/.style={draw,rectangle,minimum size=1mm,inner sep=4pt},
    scale=0.55,transform shape
   ] 
  \newcommand{\seedNode}[2]{\begin{tabular}{c}\Large \texttt{#1} \\ \hline \Large seeds: \texttt{#2}\end{tabular}}

   \node[tree] (r) {\seedNode{aabaababababaaba}{aba}}
   child {
       node[tree] {\seedNode{aabaabab}{aba}}
       child {
            node[tree] {\seedNode{aaba}{$\emptyset$}}
            child { node[tree] {\seedNode{aa}{a}} }
            child { node[tree] {\seedNode{ba}{$\emptyset$}} }
       }
       child {
            node[tree] {\seedNode{abab}{ab\textrm{,}ba}}
            child { node[tree] {\seedNode{ab}{$\emptyset$}} }
            child { node[tree] {\seedNode{ab}{$\emptyset$}} }
       }
   }
   child {
       node[tree] {\seedNode{ababaaba}{aba}}
       child {
            node[tree] {\seedNode{abab}{ab\textrm{,}ba}}
            child { node[tree] {\seedNode{ab}{$\emptyset$}} }
            child { node[tree] {\seedNode{ab}{$\emptyset$}} }
       }
       child {
            node[tree] {\seedNode{aaba}{$\emptyset$}}
            child { node[tree] {\seedNode{aa}{a}} }
            child { node[tree] {\seedNode{ba}{$\emptyset$}} }
       }
   };

\end{tikzpicture}
    \caption{A schematic view of tree $\SeedSets$ of $T$ (in the real data structure, seeds are stored on suffix trees of basic factors). For example, \texttt{ba} is a seed of $T[5 \dd 12]$ since it is a seed of basic factors $T[5 \dd 8]$ and $T[9 \dd 12]$ and its occurrence covers $T[8 \dd 9]$ (\cref{concatenation}).}\label{seedsets}
\end{figure}
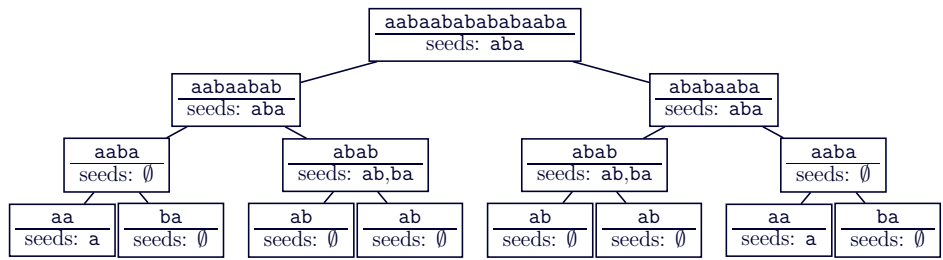

The sum of lengths of basic factors in $T$ is $\Oh(n \log n)$.
Consequently, due to~\cref{lem:seeds},
the tree $\SeedSets(T)$ has total size $\Oh(n\log n)$ and can be computed in $\Oh(n\log n)$ time. (To use \cref{lem:seeds}\eqref{xa} we renumber letters in basic factors of $T$ via bucket sort so that the letters of a basic factor $S$ are from $\{0,\ldots,|S|^{\Oh(1)}\}$.)

\section{Internal Cover of a Given Length}
In this section we show how to use $\SeedSets(T)$ to answer internal queries related to computing the longest prefix of a factor $S$ of $T$ that is covered by its length-$\ell$ prefix. We start with the following, easier queries.

\defproblemq{Cover of a Given Length Query (\IsCover($\ell,S$))}
{A factor $S$ of $T$ and a positive integer $\ell$}{
Does $S$ have a cover of length $\ell$?
}

The following three lemmas provide the building blocks of the data structure for answering $\IsCover$ queries.

\begin{lemma}[Seed of a basic factor]\label{base}
After $\Oh(n \log n)$-time preprocessing, for any factor $C$ and basic factor $B$ of $T$ such that $2|C| \le |B|$, we can check if $C$ is a seed of $B$ in $\Oh(\gamma_n)$ time.
\end{lemma}
\begin{proof}
Let $|C|=c$ and $B=T[a \dd b]$. We first ask an IPM query to find an occurrence of $C$ inside $T[a\dd a+2c-1]$. If such an occurrence does not exist, then $C$ cannot be a seed of $T[a\dd b]$ as it is already not a seed of $T[a\dd a+2c-1]$ (there must be a full occurrence to cover the middle letter, and $a+2c-1\le b$). Otherwise, we can use the occurrence to check if $C$ is a seed of $B$ with \cref{lem:seeds}\eqref{xb}.
\end{proof}

For strings $C$ and $S$, by $\Cov(C,S)$ we denote the set of positions of $S$ that are covered by occurrences of $C$.

\begin{lemma}[Covering short factors]\label{short}
After $\Oh(n)$-time preprocessing, for any two factors $C$ and $F$ of $T$ such that $|F|/|C| = \Oh(1)$, the set $\Cov(C,F)$, represented as a union of maximal intervals, can be computed in $\Oh(1)$ time.
\end{lemma}
\begin{proof}
We ask IPM queries for pattern $C$ on length-$2|C|$ factors of $F$ with step $|C|$. Each IPM query returns an arithmetic sequence of occurrences that corresponds to an interval of covered positions (possibly empty). It suffices to compute the union of these intervals.
\end{proof}

\begin{lemma}[Seeds of strings concatenation]\label{concatenation}
After $\Oh(n)$-time preprocessing, for any three factors $C$, $F_1=T[i\dd j]$ and $F_2=T[j+1\dd k]$ of $T$ such that $2|C| \le |F_1|,|F_2|$ and $C$ is a seed of both $F_1$ and $F_2$, we can check if $C$ is also a seed of $F_1F_2$ in constant time.
\end{lemma}
\begin{proof}
For a string $C$ of length $c$ being a seed of both $T[i\dd j]$ and $T[j+1\dd k]$ to be a seed of $T[i\dd k]$, it is enough if its occurrences cover the string $U=T[j-c+1\dd j+c]$. We can check this condition if we apply \cref{short} for $C$ and $F=T[j-2c+1\dd j+2c]$.
%
\end{proof}

\begin{lemma}\label{testing}
After $\Oh(n\log n)$ time and space preprocessing of $T$,
a query $\IsCover(\ell,S)$ 
can be answered in $\Oh(\log(|S|/\ell) \,\gamma_n+1)$ time.
\end{lemma}
\begin{proof}
Let $S=T[i \dd j]$, $|S|=s$ and $C=T[i\dd i+\ell-1]$.

We consider a decomposition of $S$ into basic factors, but we are only interested in basic factors of length at least $2\ell$ in the decomposition. Let $F_1,\ldots,F_k$ be those factors and $T[i \dd i'],T[j' \dd j]$ be the remaining prefix and suffix of length $\Oh(\ell)$. Note that $k=\Oh(\log(s/\ell))$. Moreover, this decomposition can be computed in $\Oh(k+1)$ time by starting from the leftmost and rightmost basic factors of length $2^b$, where $b=\ceil{\log \ell}+1$, that are contained in $S$.

If $C$ is a cover of $S$, it must be a seed of each of the basic factors $F_1,\ldots,F_k$.
We can check this condition by using \cref{base} in $\Oh(k\gamma_n)$ total time.

Next we check if $C$ is a seed of $F_1 \cdots F_k$ in $\Oh(k)$ total time using \cref{concatenation}.
Finally, we use IPM queries to check if occurrences of $C$ cover all positions in each of the strings $T[i \dd i'+c-1]$, $T[j'-c+1 \dd j]$ and if $C$ is a suffix of $T[i \dd j]$, using \cref{short}. This takes $\Oh(1)$ time.

The total time complexity is $\Oh(k\gamma_n+1)$.
\end{proof}

As we will see in the next section, $\IsCover$ queries immediately imply a slower, $\Oh(\log^2 n \,\gamma_n)$-time algorithm for answering $\shcov$ queries. However, they are also used in our algorithm for answering $\allcov$ queries. In the efficient algorithm for $\shcov$ queries we use the following generalization of $\IsCover$ queries.

\defproblemq{Longest Covered Prefix Query (\LCovP($\ell,S$))}
{A factor $S$ of $T$ and a positive integer $\ell$}{
The longest prefix $P$ of $S$ that is covered by $S[1 \dd \ell]$
}

To answer these queries, we introduce an intermediate problem that is more directly related to the range tree containing seeds representations.

\defproblem{\SeededPref($C,\ell,S)$ query} 
{A length-$\ell$ factor $C$ of $T$ and a factor $S$ being a concatenation of basic factors of $T$ of length $2^p$, where $p=\min\{q\in\mathbb{Z}\,:\,2^q \ge 2\ell\}$}
{The length $m$ of the longest prefix of $S$ which is a concatenation of basic factors of length $2^p$ such that $C$ is a seed of this prefix}

In other words, we consider only blocks of $S$ which are basic factors of length $2^p=\Theta(\ell)$.
Everything starts and ends in the beginning/end of a basic factor of length $2^p$.
The number of such blocks in the prefix returned by $\SeededPref$ is $\Oh(\result'/\ell)$, where $\result'=\SeededPref(C,\ell,S)$, and, as we show in \cref{LSeededPref}, it can be computed in $\Oh(\log(\result'/\ell)\gamma_n+1)$ time. This is how we achieve $\Oh(\log (\result/\ell)\,\gamma_n+1)$ time for $\LCovP(\ell,S)$ queries.
In a certain sense the computations behind \cref{LSeededPref} can work in a pruned range tree $\SeedSets(T)$.
\begin{lemma}\label{LRed} 
%
%
After $\Oh(n)$-time preprocessing, a $\LCovP(\ell,S)$ query reduces in $\Oh(1)$ time to a \linebreak $\SeededPref(C,\ell,S')$ query with $|S'|\le |S|$.
\end{lemma}
\begin{proof}
First, let us check if the answer to $\LCovP(\ell,S)$ is small, i.e.\ at most $4\ell$, using \cref{short}. Otherwise, let $p$ be defined as in a $\SeededPref$ query, $C=S[1 \dd \ell]$ and $S'$ be the maximal factor of $S$ that is composed of basic factors of length $2^p$ ($S'$ can be the empty string, if $|S| < 3\cdot 2^p$). Let $S=T[i \dd j]$ and $S'=T[i' \dd j']$. Then
$$|(i'+\SeededPref(C,\ell,S'))-(i+\LCovP(\ell,S))|< 2^p;$$
see \cref{Reduction}.
Hence, knowing $d=\SeededPref(C,\ell,S')$, we check in $\Oh(1)$ time, using \cref{short} in a factor $T[i'+d-2^p \dd i'+d+2^p-1]$ of length $2^{p+1}$, what is the exact value of $\LCovP(\ell,S)$.

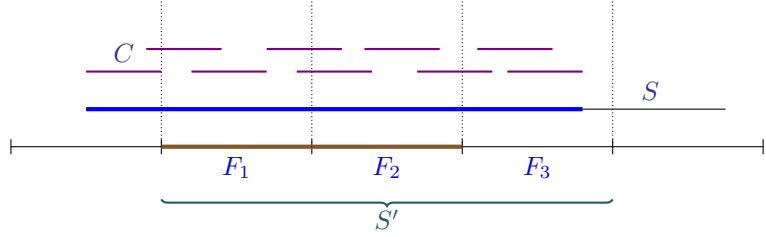
\begin{figure}
    \centering
\begin{tikzpicture}
    \draw (0,0) -- (10,0);
    \foreach \x in {0,2,...,10}{\draw (\x,-0.1) -- (\x,0.1);}
    \draw[brown!70!black, ultra thick] (2,0) -- (6,0);
    \draw[thick,darkteal,snake=brace] (8,-0.7) -- node[below] {$S'$} (2,-0.7);
    \draw[thick,blue] (3,0) node[below] {$F_1$};
    \draw[thick,blue] (5,0) node[below] {$F_2$};
    \draw[thick,blue] (7,0) node[below] {$F_3$};
    \draw (1,0.5) -- (9.5,0.5);
    \draw[darkblue] (8.5,0.5) node[above] {$S$};
    \draw[blue,ultra thick] (1,0.5) -- (7.6,0.5);
    \foreach \x/\y in {1/1,1.8/1.3,2.4/1,3.4/1.3,3.8/1,4.7/1.3,5.4/1,6.2/1.3,6.6/1}{
        \draw[thick,violet,xshift=\x cm,yshift=\y cm] (0,0) -- (1,0);
    }
    \draw[darkblue,thick] (1.5,1) node[above] {$C$};
    \foreach \x in {2,4,6,8}{\draw[densely dotted] (\x,0) -- (\x,2);}
\end{tikzpicture}
    \caption{$F_1$, $F_2$, $F_3$ are basic factors of length $2^p$. The answers to $\LCovP(\ell,S)$ and $\SeededPref(C,\ell,S')$ queries are shown in bold. Note that $C$ is a seed of $F_1$ and $F_2$ and that it could be the case that $C$ is also a seed of $F_3$, even though it has no further full occurrence.}
    \label{Reduction}
\end{figure}

We compute $p$ using the formula $p=1+\ceil{\log \ell}$. Then the endpoints of $S'$ can be computed from the endpoints of $S$ in $\Oh(1)$ time using simple modular arithmetic. The $\Oh(n)$ preprocessing is due to \cref{short}.
\end{proof}

\newcommand{\SeededBasic}{\mathit{SeededBasic}}
\newcommand{\TestConcat}{\mathit{TestConcat}}
\newcommand{\TestExtend}{\mathit{TestExtend}}

To answer $\SeededPref$ queries we use our range tree which stores seeds of every basic factor.
Recall that for each basic factor $T[i \dd j]$ we can check if $C$ is a seed of this factor in $\Oh(\gamma_n)$ time (\cref{base}); we denote this test $\SeededBasic(C,i,j)$. 

Also for any two neighboring factors $T[i \dd j]$, $T[j+1 \dd k]$, for which $C$ is a seed, we can check in $\Oh(1)$ time if $C$ is a seed of the composite factor $T[i \dd k]$ (\cref{concatenation}); we denote this test $\TestConcat(C,i,j,k)$.

\newcommand{\rank}{\mathit{rank}}
\begin{lemma}\label{LSeededPref}
After $\Oh(n \log n)$ time and space preprocessing of $T$, a query $\SeededPref(C,\ell,S)$ 
can be answered in $\Oh(\log(\result/\ell) \,\gamma_n+1)$ time, where $\result=|\SeededPref(C,\ell,S)|$.
\end{lemma}
\begin{proof}
Let us define 
$$\rank(i)\,=\, \max\,\{\, k\;:\;  [i\dd i+2^k)\ \mbox{is a basic interval}\;\}.$$
All $\rank$ values for $i=1,\ldots,n$ can be computed in $\Oh(n)$ time from the basic intervals.
\begin{observation}
If $\rank(i)\ge k$, then $\rank(i+2^k)\ge k$.
\end{observation}

We introduce a Boolean function that is applied only if $\rank(j+1) \ge k$:
 $$\TestExtend(i,j,k) \Leftrightarrow \SeededBasic(C,j+1,j+2^k) \land \TestConcat(C,i,j,j+2^k).$$

We can then use the following \cref{seededprefalg} to compute the result of a query. The algorithm implicitly traverses the static range tree. For an illustration, see \cref{bigdrzewo}, where the
\emph{Doubling Phase} corresponds to ascending the tree, 
and the \emph{Binary Search Phase} corresponds to descending the tree. Intuitively, if $\rank(i) \ge k$, then the basic interval $[i \dd i+2^k)$ is the left child of its parent in the tree if and only if $\rank(i)>k$.

Variable $k$ is incremented in every second step of the Doubling Phase. At the conclusion of the phase, we know that $C$ is a seed of $T[start \dd last]$, where $last-start+1 \ge 2^k$, and the final output will be $T[start \dd last']$, where $(last'-last) \in [0 \dd 2^k)$. Intuitively, we already have an approximation and use the Binary Search Phase to compute the actual result.

\SetKwBlock{Repeat}{repeat}{}
\begin{algorithm}
\tcp{Doubling Phase:}
Let $S=T[start\dd end]$ and $p=1+\ceil{\log \ell}$\;
$last:=start+2^p-1;$ $k:=p$\;
\Repeat{
      \tcp{Invariant: $C$ is a seed of $T[start \dd last]$, $last-start+1 \ge 2^k$ and $\rank(last+1) \ge k$.}
      \If{$last \ge end$}{
        \Return{$S$}\;
      }
      \lIf{\KwSty{not} $\TestExtend(start,last,k)$}{\KwSty{break}}
      $last:=last+2^k$\;
      \lIf{$\rank(last+1)>k$}{$k:=k+1$}
}
~\\
\tcp{Binary Search Phase:}
\Repeat{
      $k:=k-1$\;
      \lIf{$k<p$}{\KwSty{break}}
      \If{$last \ge end$}{\Return{$S$}\;}
      \If{$\TestExtend(start,last,k)$}{$last:=last+2^k$\;}
}
\Return{$T[start\dd last]$}\;
\caption{Compute $\SeededPref(C,\ell,S)$}\label{seededprefalg}
\end{algorithm}

 \begin{figure}[htpb]
 \centering
\begin{tikzpicture}[scale=1.05,-, >=stealth, auto, semithick,
    level/.style={sibling distance=80mm/#1},
    level 1/.style={sibling distance=110mm},
    level 2/.style={sibling distance=55mm},
    level 3/.style={sibling distance=25mm},
    level 4/.style={sibling distance=15mm},
    level 5/.style={sibling distance=7mm},
    tree/.style={draw,circle,fill=gray,minimum size=1mm,inner sep=0pt},
    red arrow/.style={red,thick,->},
    blue arrow/.style={blue,thick,->},
    scale=0.48, transform shape,
   ] 

   \node[tree] (r) {}
     child {node[tree]  {}
       child {node[tree]  {}
         child {node[tree] {}
           child {node[tree] {}
             child {node[tree] {}}
             child {node[tree] {}}
           }
           child {node[red,tree] (p2) {}
             child {node[red,tree] (p1) {}}
             child {node[tree] {}}
           }
         }
         child { node[red,tree] (p4) {}
           child {node[red,tree] (p3) {}
             child {node[tree] {}}
             child {node[tree] {}}
           }
           child {node[tree] {}
             child {node[tree] {}}
             child {node[tree] {}}
           }
         }
       }
       child {node[red,tree] (p6)  {}
         child { node[red,tree] (p5) {}
           child {node[tree] {}
             child {node[tree] {}}
             child {node[tree] {}}
           }
           child {node[tree] {}
             child {node[tree] {}}
             child {node[tree] {}}
           }
         }
         child { node[tree] {}
           child {node[tree] {}
             child {node[tree] {}}
             child {node[tree] {}}
           }
           child {node[tree] {}
             child {node[tree] {}}
             child {node[tree] {}}
           }
         }
       }
     }
     child {node [black,tree] (p8) {}
       child {node [red,tree] (p7)  {}
         child { node[tree] {}
           child {node[tree] {}
             child {node[tree] {}}
             child {node[tree] {}}
           }
           child {node[tree] {}
             child {node[tree] {}}
             child {node[tree] {}}
           }
         }
         child { node[tree] {}
           child {node[tree] {}
             child {node[tree] {}}
             child {node[tree] {}}
           }
           child {node[tree] {}
             child {node[tree] {}}
            child {node[tree] {}}
            }
         }
       }
     child {node[blue,tree] (p9) {}
         child { node [blue,tree] (p10) {}
           child {node[tree] {}
             child {node[tree] {}}
             child {node[tree] {}}
           }
           child {node[tree]  {}
             child {node[tree] {}}
             child {node[tree] {}}
           }
         }
         child { node [blue,tree] (p11) {}
           child {node [blue,tree] (p12) {}
             child {node [blue,tree] (p13) {}}
             child {node[tree] {}}
           }
           child {node[tree] {}
             child {node[tree] {}}
             child {node[tree] {}}
           }
         }
     }
   };

   \draw[red arrow] (p1)--(p2);
   \draw[red arrow] (p2)--(p3);
   \draw[red arrow] (p3)--(p4);
   \draw[red arrow] (p4)--(p5);
   \draw[red arrow] (p5)--(p6);
   \draw[red arrow] (p6)--(p7);
   \draw[red arrow] (p7)--(p8);
   \draw[blue arrow] (p8)--(p9);
   \draw[blue arrow] (p9)--(p10);
   \draw[blue arrow] (p10)--(p11);
   \draw[blue arrow] (p11)--(p12);
   \draw[blue arrow] (p12)--(p13);

   \foreach \i/\y in {1/0,2/0,3/10pt,4/0,5/10pt,6/0,7/10pt} {
     \node [left=of p\i,xshift=25pt,yshift=\y,fill=none, red,draw,circle,minimum size=1mm,inner sep=0pt] {\Large $+$};
   }

   \foreach \i/\y in {8/0, 9/0, 11/0, 12/0} {
     \node [right=of p\i,xshift=-25pt,yshift=\y,fill=none, blue,draw,circle,minimum size=1mm,inner sep=0pt] {\Large $-$};
   }
   \foreach \i/\y in {10/0} {
     \node [left=of p\i,xshift=25pt,yshift=\y,fill=none, blue,draw,circle,minimum size=1mm,inner sep=0pt] {\Large $+$};
   }
   \foreach \i/\y in {13/0} {
     \node [right=of p\i,xshift=-25pt,yshift=\y,fill=none, blue,draw,circle,minimum size=1mm,inner sep=0pt] {\Large $+$};
   }

   \draw[thick,violet,decorate,decoration={brace,raise=5pt,amplitude=5pt,mirror}] 
         (p1.center)--(p13.center) node[midway,fill=none,below,yshift=-20pt,line width=0pt] {\huge result};

  \draw[red,very thick,->] (-11cm,-6cm)--+(0,2cm) node[above] {\huge ascend};
  \draw[blue,very thick,->] (11cm,-4cm) node[above] {\huge descend} --+(0,-2cm) ;

  \end{tikzpicture}
 \caption{Interpretation of \cref{seededprefalg} on the range tree. Basic factors correspond to the nodes of the tree.
 If the basic interval in the currently queried node corresponds to $S[i\dd j]$,
 then in this moment we know that $C$ is a seed of $S[1\dd i-1]$. The query asks whether $C$ is a seed of the basic factor $S[i\dd j]$
 (in $\Oh(\gamma_n)$ time) and  whether the concatenation of $S[1\dd i-1]$ and this basic factor is seeded by $C$ (constant time). 
 If "yes", then    
 the known seeded prefix is extended and ends at $j$.}\label{bigdrzewo}
 \end{figure}
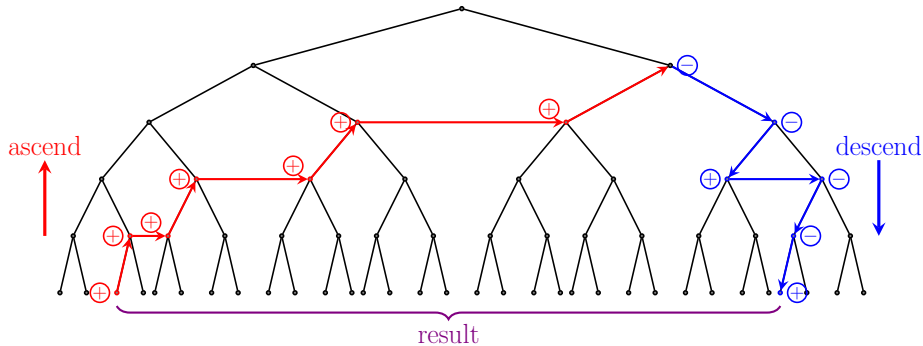
Each of the phases makes at most $\Oh(\log (\result / \ell))$ iterations and uses $\Oh(\gamma_n)$ time for each iteration. Thus, we have arrived at the required complexity.
\end{proof}

As a corollary of \cref{LRed,LSeededPref}, we obtain the following result.

\begin{lemma}\label{LCP}
After $\Oh(n \log n)$ time and space preprocessing of $T$, a query $\LCovP(\ell,S)$ 
can be answered in $\Oh(\log(\result/\ell) \,\gamma_n+1)$ time, where $\result=|\LCovP(\ell,S)|$.
\end{lemma}

\section{Internal Shortest Cover Queries}
For a string $S$, by $\Borders(S)$ we denote a decomposition of the set of all border lengths of $S$ into $\Oh(\log |S|)$ arithmetic sequences $A_1,\ldots,A_k$ such that each sequence $A_i$ is either a singleton or, if $p$ is its difference, then the borders with lengths in $A_i\setminus \{\min(A_i)\}$ are periodic with the shortest period $p$. Moreover, $\max(A_i) < \min(A_{i+1})$ for every $i \in [1\dd k-1]$. See e.g.~\cite{DBLP:conf/cpm/CrochemoreIKKRRTW12}. The following lemma is shown by applying a period query (\cref{PQ}).

\begin{lemma}[\cite{tomeksthesis,DBLP:conf/soda/KociumakaRRW15}]\label{borders}
For any factor $S$ of $T$, $\Borders(S)$ can be computed in $\Oh(\log n)$ time after $\Oh(n)$-time preprocessing.
\end{lemma}

\subsection{Simple Algorithm with $\Oh(\log^2 n\, \gamma_n)$ Query Time}
Let us start with a much simpler but slower algorithm for answering $\shcov$ queries using $\IsCover$ queries. We improve it in \cref{shortest} by using $\LCovP$ queries and applying an algorithm for computing shortest covers that resembles, to some extent, computation of the shortest cover from~\cite{DBLP:journals/ipl/ApostolicoFI91}.

\begin{proposition}
Let $T$ be a string of length $n$. After $\Oh(n \log n)$-time preprocessing, for any factor $S$ of $T$ we can answer a $\shcov(S)$ query in $\Oh(\log^2 n \log \log n)$ time.
\end{proposition}
\begin{proof}
Using \cref{borders} we compute the set $\Borders(S)=A_1,\ldots,A_k$ in $\Oh(\log n)$ time.
Let us observe that the shortest cover of a string is aperiodic.
 This implies that from each progression $A_i$ only the border of length $\min(A_i)$ can be the shortest cover of $S$.
 We use \cref{testing} to test each of the $\Oh(\log n)$ candidates in $\Oh(\log n \, \gamma_n)$ time.
\end{proof}

\subsection{Faster Queries}
\begin{theorem}\label{shortest}
Let $T$ be a string of length $n$. After $\Oh(n \log n)$-time preprocessing, for any factor $S$ of $T$ we can answer a $\shcov(S)$ query in $\Oh(\log n \log\log n)$ time.
\end{theorem}
\begin{proof}
Again we use \cref{borders} we compute the set $\Borders(S)=A_1,\ldots,A_k$,
 in $\Oh(\log n)$ time.
Let us denote the border of length $\min(A_i)$ by $C_i$ and $C_{k+1}=S$. We assume that $C_i$'s are 
sorted in increasing order of lengths. 
Then we proceed as shown in \cref{algo_shcov}. See also Fig.~\ref{schodki}.

\medskip
\begin{algorithm}[H]
$i:=1$\;
\While{\KwSty{true}}{
    \tcp{Invariant: $C_1,\ldots,C_{i-1}$ are not covers of $S$}
    \tcp{$C_i$ is an {\it active} border}
    $P:=\textsc{\LCovP}(|C_i|,S)$\;
    \lIf{$P=S$}{\Return{$|C_i|$}}
    \While{$|C_i| \le |P|$}{
        $i:=i+1$\;
    }
}
\caption{$\shcov(S)$ query.}\label{algo_shcov}
\end{algorithm}

\medskip
To argue for the correctness of the algorithm it suffices to show the invariant. The proof goes by induction. 

The base case is trivial. Let us consider the value of $i$ at the beginning of a step of the while-loop. If $P=S$, then by the inductive assumption $C_i$ is the shortest cover of $S$ and can be returned. Otherwise, $C_i$ is not a cover of $S$. 

Moreover, for each $j$ such that $|C_i| < |C_j| \le |P|$, since $C_j$ is a prefix of $P$, $C_i$ is a seed of $C_j$. Moreover, both $C_i$ and $C_j$ are borders of $S$, so $C_i$ is a border of $C_j$. Consequently, $C_j$ cannot be a cover of $T$, as then $C_i$ would also be a cover of $T$ by Observation~\ref{obsbasic}. This shows that the inner while-loop correctly increases $i$.

The algorithm stops because at each point $|P| \ge |C_i|$ and $i$ is increased.

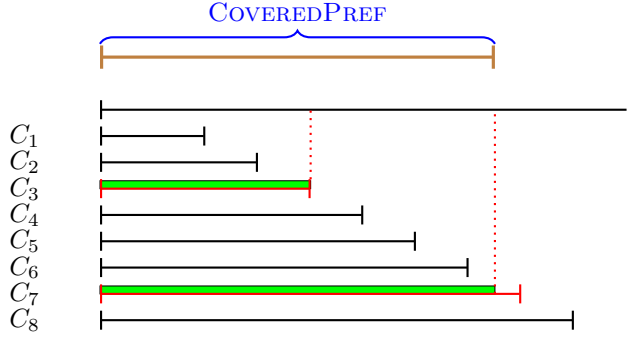
\begin{figure}
    \centering
\begin{tikzpicture}[scale=0.7]
    \draw[thick,|-] (0, 0)--+(10,0);

    \draw[very thick,brown,|-|] (0, 1)--+(7.5,0);
    \draw[thick,blue,decorate,decoration={brace,raise=5pt,amplitude=5pt}] 
      (0,1)--(7.5,1) node[midway,above,yshift=10pt] {\LCovP};

    \draw[fill=green] (0,-1.5) rectangle +(4, 0.15);    
    \draw[fill=green] (0,-3.5) rectangle +(7.5, 0.15);    

    \foreach[count=\y] \label/\x/\style in {1/2/, 2/3/, 3/4/red, 4/5/, 5/6/, 6/7/, 7/8/red, 8/9/} {
      \node (l\y) at (0, -\y*0.5) {};
      \node (r\y) at (\x, -\y*0.5) {};
      \node [left of=l\y] {$C_\label$};
      \draw[|-|,thick,\style] (l\y.center)--(r\y.center);
    }

    \draw[thick,dotted,red] (r3)--+(0,1.5);
    \draw[thick,dotted,red] (7.5,-3.5)--(7.5,0);
\end{tikzpicture}
    \caption{If $C_3$ is an active border, then the next active one is $C_7$. We skip $C_4,C_5,C_6$ as candidates for the shortest cover.}\label{schodki}
    \label{fig:my_label}
\end{figure}

Let $c_1,\ldots,c_p$ be equal to the length of an active border 
in the algorithm at the start of subsequent outer while-loop iterations and let $c_{p+1}=|S|$. 

Let us note that, for all $j=1,\ldots,p$, $|\textsc{\LCovP}(c_j,S)| \le c_{j+1}$. By \cref{LCP}, the total complexity of answering longest covered prefix queries in the algorithm is at most
$$
    \Oh\left(p+\gamma_n \sum_{j=1}^{p} \log \tfrac{c_{j+1}}{c_j} \right) 
    =
    \Oh(\log n + \gamma_n (\log c_{p+1} - \log c_1)) = \Oh(\log n\,\gamma_n).
$$

The preprocessing of Lemmas~\ref{LCP} and \ref{borders} takes $\Oh(n \log n)$ time. The conclusion follows.
\end{proof}

If $\shcov$ queries are to be answered in a batch, we can use off-line WA queries of \cref{WAQ} to save the $\gamma_n$-factor. We can also avoid storing the whole data structure $\SeedSets$ by using an approximate version of $\LCovP$ queries.

\begin{theorem}\label{corMinOffline}
For a string $T$ of length $n$, any $m$ queries $\shcov(T[i\dd j])$ can be answered in $\Oh((n+m) \log n)$ time and $\Oh(n+m)$ space.
\end{theorem}
\begin{proof}
The $\,\gamma_n$ factor from the query complexity of \cref{shortest} stems from using 
on-line weighted-ancestor queries to access $\SeedSets$.
In the off-line setting it suffices to answer such queries in a batch; see \cref{WAQ}.

The only data structure from \cref{shortest} which takes $\omega(n)$ space is the $\SeedSets$ tree.
However each level of the tree, corresponding to basic factors of the same length, takes
only $\Oh(n)$ space and can be constructed independently in $\Oh(n)$ time. We will modify the algorithm for answering a $\shcov$ query so that it will access the levels of $\SeedSets$ in the order of increasing lengths of basic factors. Then we will be able to answer all $\shcov$ queries simultaneously in $\Oh(n+m)$ space.

\paragraph{\bf Approximate $\LCovP$ queries.} The building block of the data structure for answering $\shcov$ queries are $\LCovP$ queries. In \cref{LRed}, a $\LCovP(\ell,S)$ query is either answered in $\Oh(1)$ time via IPM queries, or reduced to a $\SeededPref(C,\ell,S')$ query for $S'$ being a maximal factor of $S$ that is a concatenation of basic factors of length $\Delta \ge \ceil{2\ell}$ and $C=S[1 \dd \ell]$.

Let $\result'=|\SeededPref(C,\ell,S')|$. 
Note that if $\result' < 5\Delta$, then a $\SeededPref$ query can be answered naively via IPM queries in $\Oh(1)$ time. Otherwise, the algorithm for answering $\SeededPref$ queries (\cref{LSeededPref}) consists of two phases; the first phase considers basic factors of non-decreasing lengths, but the second phase considers them according to non-increasing lengths.

This implementation does not satisfy the condition that the algorithm visits $\SeedSets$ level by level. However, we will show that the result of the first phase yields a constant factor approximation of the $\result=|\LCovP(\ell,S)|$.
Indeed, 
after the first phase a prefix $U$ of $S'$ is computed such that $C$ is a seed of $U$ and $|U| \ge \frac12 \result'$. If $S=LS'R$, this means that $C$ is a cover of a prefix $P$ of $LU$ of length at least $|L|+\frac12 \result'-(\ell-1)$. 

We obviously have $|P| \le \result$. Moreover, we have $\result \le |L| + \result' + 2\ell-1$. Hence,
$$3|P| > |L|+\tfrac32\result'-3\ell + 1 \ge |L|+\result' + 5\ell - 3\ell+1 \ge \result,$$
so indeed $|P|$ is a 3-approximation of $\result$. Let us call the resulting routine \textsc{ApproxCoveredPref}.

\paragraph{\bf Simultaneous calls to $\IsCover$.} We use the approximate routine instead of $\LCovP$ to compute $P$ in \cref{algo_shcov}. Then all candidates $C_i$ that are eliminated in the inner while-loop are eliminated correctly. The only issue is with correctness of the if-statement, since $P$ is only a lower bound for the result. To address this issue, for each active border in the algorithm we start running $\IsCover(|C_i|,S)$. According to \cref{testing}, this requires $\Oh(1)$-time checks using IPM queries on the edges of $S$ and checking if $C_i$ is a seed of concatenation of basic factors $F_1,\ldots,F_k$, each of length at least $2\ell$. To this end, we use \cref{base,concatenation}.

Let us recall that the sequence of lengths of basic factors $F_1,\ldots,F_k$ is first increasing and then decreasing, and it can be computed step by step in $\Oh(k)$ time. Hence, we can answer queries for $F_i$ starting from the ends of the sequence simultaneously with computing \textsc{ApproxCoveredPref}.

At the conclusion of the latter query, if $|P| < \frac13 |S|$, then we know that $\IsCover(|C_i|,S)$ would return false and we can discard the computations. Otherwise we continue computing $\IsCover(|C_i|,S)$ for this active border in the subsequent steps of the outer while-loop. This may yield several $\IsCover$ queries that are to be answered in parallel. However, the number of such queries is only $\Oh(1)$, since $|C_{i+2}| > \frac32 |C_i|$ (otherwise $C_{i+1}$ would have been in the same arithmetic sequence as one of $C_i$, $C_{i+2}$).
\end{proof}

\section{Internal All Covers Queries}
In this section we refer to $\allcov(S)$ as to the set of lengths of all covers of $S$.
This set consists of a logarithmic number of arithmetic sequences since the same is true for all borders.
In each sequence of borders we show that it is needed only to check $\Oh(1)$ borders to be a cover of $S$.
Hence we start with an algorithm testing any sequence of $\Oh(\log n)$ candidate borders.
\subsection{Verifying $\Oh(\log n)$ Candidates}
Assume that $B$ is an increasing sequence $b_1,\dd,b_k$ of lengths of borders of a given factor $S$ (not necessarily all borders), with $b_k=|S|$.
A {\it chain} in $B$ is a maximal subsequence $b_i,\ldots,b_j$  of consecutive elements of $B$ such that $S[1\dd b_t]$ is a cover of $S[1\dd b_{t+1}]$ for each $t\in [i\dd j)$.
From Observation~\ref{obsbasic} we get the following.
\begin{observation}\label{graph_properties}
The set of elements of a chain that belong to $\allcov(S)$ is a prefix of this chain.
Moreover, if the last element of a chain is not $|S|$, then it is not a cover of $S$.
\end{observation}
\newcommand{\chains}{\mathit{chains}}
\newcommand{\covers}{\mathit{covers}}
\newcommand{\refine}{\mathit{refine}}
\newcommand{\computeUsing}{\mathit{computeUsing}}
\newcommand{\Prev}{\mathit{prev}}
\newcommand{\Next}{\mathit{next}}
 We denote by $\chains(B)$ and $\covers(B)$, respectively, 
 the partition of $B$ into chains and the set of elements $b\in B$ such that $S[1\dd b]$ is a cover of $S$.
 For $b \in B$ by $\Prev(b)$ we denote the previous element in its chain (if it exists). Moreover, for $C \subseteq B$ by $\Next_C(b)$ we denote the smallest $c \in C$ such that $c>b$.

\begin{lemma}\label{candidates}
Let $T$ be a string of length $n$. After $\Oh(n \log n)$-time preprocessing, for any factor $S$ of $T$ and a sequence $B$ of $\Oh(\log n)$ borders of $S$ we can compute $\covers(B)$ in $\Oh(\log n \log \log n\,\gamma_n)$ time.
\end{lemma}
\begin{proof} We introduce two operations and use them in a recursive \cref{algo_allcov}.

\begin{description}
\item{$\mathbf{refine}(B)$:}  
%
 removes the last element of each chain in $B$ and every second element of each chain, except $|S|$ (see \cref{refine}). Note that $|\refine(B)| \le |B|/2+1$.
 \vskip 0.1cm
\item{$\mathbf{computeUsing}(B,C)$:}
%
Assuming that we know the set $C$ of all covers of $S$ among $\refine(B)$, for each element $b$ of $B\setminus \refine(B)$ we add it to $C$ if
$\Prev(b) \in C$ and $S[1 \dd b]$ is a cover of $S[1 \dd \Next_C(b)]$.
The set of all elements that satisfy this condition together with $C$ is returned as $\covers(B)$.
\end{description}

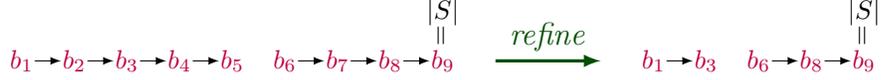
\begin{figure}
    \centering
\begin{tikzpicture}[scale=0.7]
    \foreach \x/\c in {1/1,2/2,3/3,4/4,5/5,6/6,7/7,8/8,9/9}{
        \draw[purple] (\x,0) node {$b_\c$};
    }
    \draw (9,0.5) node[above] {$|S|$};
    \draw (9,0.85) node[below] {$\mathrel{\rotatebox{90}{$=$}}$};
    \foreach \x in {1,2,3,4,6,7,8}{
        \draw[xshift=\x cm,-latex] (0.25,0) -- (0.75,0);
    }
    \draw[green!30!black,very thick,-latex] (10,0) -- node[above] {\large $\refine$} (12,0);
    \begin{scope}[xshift=12cm]
    \foreach \x/\c in {1/1,2/3,3/6,4/8,5/9}{
        \draw[purple] (\x,0) node {$b_\c$};
    }
    \draw (5,0.5) node[above] {$|S|$};
    \draw (5,0.85) node[below] {$\mathrel{\rotatebox{90}{$=$}}$};
    \foreach \x in {1,3,4}{
        \draw[xshift=\x cm,-latex] (0.25,0) -- (0.75,0);
    }
    \end{scope}
\end{tikzpicture}
    \caption{There is an arrow from $b_i$ to $b_{i+1}$ iff $S[1 \dd b_i]$ is a cover of $S[1 \dd b_{i+1}]$. Note that all elements in the last chain $b_6,b_7,b_8,b_9$ are cover lengths of $S$, $b_5$ is not, but some prefix of $b_1,b_2,b_3,b_4$ may be.}
    \label{refine}
\end{figure}

\begin{algorithm}[H]
Compute $\chains(B)$\;
\lIf{$B$ is a single chain (ending with $|S|$)}{\Return{$B$}}
$B' := \refine(B) $\tcp*{$|B'| \le |B|/2+1$}
$C:=\covers(B')$\;
\Return{$\computeUsing(B,C)$\;}
\caption{$\covers(B)$}\label{algo_allcov}
\end{algorithm}

\medskip
If $B=(b_1,\dd,b_k)$, then $\chains(B)$ can be constructed in $\Oh(\sum_{i=1}^{k-1}(\log{\frac{b_{i+1}}{b_i}}\,\gamma_n+1))=\Oh(\log n\, \gamma_n)$ time using \cref{testing}. Similarly, operation $\computeUsing(B,C)$ requires $\Oh(\log n\, \gamma_n)$ time since the intervals $[b,\Next_C(b)]$ for $b \in B \setminus \refine(B)$ such that $\Prev(b) \in C$ are pairwise disjoint.
The depth of recursion of \cref{algo_allcov} is $\Oh(\log \log n)$. This implies the required complexity.
\end{proof}

\subsection{Computing Periodic Covers}
Our tool for periodic covers are (as usual) {\it runs}. 
A \emph{run} (also known as a \emph{maximal repetition}) is a periodic factor $R=T[a\dd b]$ which can be extended neither to the left nor to the right without increasing the period $p=\per(R)$, i.e.,
$T[a-1]\neq T[a +p-1]$ and $T[b-p+1] \neq T[b+1]$
provided that the respective positions exist.
The following observation is well-known.

\begin{observation}\label{p_1}
Two runs in $T$ with the same period $p$ can overlap on at most $p-1$ positions.
\end{observation}

The \emph{exponent} $\expon(S)$ of a string $S$ is $|S|/\per(S)$.
The \emph{Lyndon root} of a string $S$ is the minimal cyclic shift of $S[1\dd \per(S)]$.

If $S=T[a \dd b]$ is periodic, then by $\run(S)$ we denote the run $R$ with the same period that contains $S$. We say that $S$ is \emph{induced} by $R$. A periodic factor of $T$ is induced by exactly one run~\cite{DBLP:journals/tcs/CrochemoreIKRRW14}. The $\run$-queries are essentially equivalent to two-period queries. By $\mathcal{R}(T)$ we denote the set of all runs in a string $T$.

\begin{lemma}[\cite{DBLP:journals/siamcomp/BannaiIINTT17,DBLP:journals/tcs/CrochemoreIKRRW14,DBLP:conf/focs/KolpakovK99}]\label{allRuns}
\begin{enumerate}[(a)]
\item $|\mathcal{R}(T)|\leq n$ and $\mathcal{R}(T)$ can be computed in $\cO(n)$ time.
\item After $\Oh(n)$-time preprocessing, $\run(S)$ queries can be answered in $\Oh(1)$ time.
\item The runs from $\mathcal{R}(T)$ can be grouped by their Lyndon roots in $\Oh(n)$ time.
\end{enumerate}
\end{lemma}

The following lemma implies that indeed for any string $S$, $\allcov(S)$ can be expressed as a union of $\Oh(\log |S|)$ arithmetic sequences. It also shows a relation between periodic covers and runs in $S$.
\begin{lemma}\label{perComb}
Let $S$ be a string, $A \in \Borders(S)$ be an arithmetic sequence with difference $p$, $A'=A \setminus \{\min(A)\}$ and $a' = \min(A')$. Moreover, let $x$ be the minimal exponent of a run in $S$ with Lyndon root being a cyclic shift of $S[1 \dd p]$.
\begin{enumerate}[(a)]
    \item\label{ita} If $a' \not\in \C(S)$, then $A' \cap \C(S) = \emptyset$.\vskip 0.2cm
    \item\label{itb} Otherwise, there exists $c \in ((x-2)p,xp] \cap A'$ such that
    $A' \cap \C(S) = \{a',a'+p,\ldots,c\}$.
\end{enumerate}
\end{lemma}
\begin{proof}
Part~\eqref{ita} follows from Observation~\ref{obsbasic}. Indeed, assume that $S$ has a cover of length $b \in A'$, with $b>a'$. As $S[1 \dd a']$ is a cover of $S[1 \dd b]$, we would have $a' \in \C(S)$.

We proceed to the proof of part~\eqref{itb}. Let $c$ be the maximum element of $A'$ such that $C:=S[1 \dd c]$ is a cover of $S$. By the same argument as before, we have that $A' \cap \C(S) = A' \cap [1,c]$. It suffices to prove the bounds for $c$.

Let $L$ be the minimum cyclic shift of $S[1 \dd p]$. We consider all runs $R_1,\ldots,R_k$ in $S$ with Lyndon root $L$. Each occurrence of $C$ in $S$ is induced by one of them. Each of the runs must hold an occurrence of $C$. Indeed, by Observation~\ref{p_1}, no two of the runs overlap on more than $p-1$ positions, so the $p$th position of each run cannot be covered by occurrences of $C$ that are induced by other runs. The shortest of the runs has length $xp$, so $c \le xp$. 

Furthermore, let $C'=S[1 \dd c']$ be a prefix of $S$ of length $c'=c+p$. If $p\cdot \expon(R_i) \ge c'+p-1$, then $R_i$ induces an occurrence of $C'$ and $\Cov(C',R_i) = \Cov(C,R_i)$. Hence, if $px \ge c'+p-1$ would hold, $C'$ would be a cover of $S$, which contradicts our assumption. Therefore, $px < c'+p-1 = c+2p-1$, so $c>(x-2)p$.
\end{proof}

\noindent
\cref{periodic} transforms \cref{perComb} into a data structure. We use static dictionaries.

\begin{lemma}[Ružić~\cite{DBLP:conf/icalp/Ruzic08}]\label{Ruzic}
A static dictionary of $n$ integers that supports $\Oh(1)$-time lookups can be stored in $\Oh(n)$ space and constructed in $\Oh(n (\log \log n)^2)$ time. The elements stored in the dictionary may be accompanied by satellite data.
\end{lemma}

\begin{lemma}[Computing $\Oh(\log n)$ Candidates]\label{periodic}\mbox{ \ }\\
For any factor $S$ of $T$ we can compute in $\Oh(\log n )$ time $\Oh(\log n)$
borders of $S$ which are \mbox{\sl candidates} for covers of $S$.  
After knowing which of these candidates are covers of $S$, we can in $\Oh(\log n)$ time represent (as $\Oh(\log n)$ arithmetic 
progressions) all borders which are covers of $S$.
The preprocessing time is $\Oh(n (\log \log n)^2)$ and the space used is $\Oh(n)$.
\end{lemma}
\begin{proof} 
It is enough to show that 
for any factor $S$ of $T$ and a single arithmetic sequence $A \in \Borders(S)$
we can compute in $\Oh(1)$ time up to four candidate borders.
Then, after knowing which of them are covers of $S$, we can in $\Oh(1)$ time represent (as a prefix subsequence of $A$) all borders in $A$ 
which are covers of $S$.
We first describe the data structure and then the query algorithm.

\smallskip\noindent
\textbf{Data structure.}
Let $T[a_1 \dd b_1],\ldots,T[a_k \dd b_k]$ be the set of all runs in $T$ with Lyndon root $L$, with $a_1 < \dots < a_k$ (and $b_1 < \dots < b_k$). The part of the data structure for this Lyndon root consists of an array $A_L$ containing $a_1,\ldots,a_k$, an array $E_L$ containing the exponents of the respective runs, as well as a dictionary on $A_L$ and a range-minimum query data structure on $E_L$. Formally, to each Lyndon root we assign an integer identifier in $[1,n]$ that is retained with every run with this Lyndon root and use it to index the data structures. We also store a dictionary of all the runs. The data structure takes $\Oh(n)$ space and can be constructed in $\Oh(n(\log \log n)^2)$ time by \cref{RMQ,Ruzic,allRuns} (RMQ, computing runs and grouping runs by Lyndon roots, and static dictionary, respectively).
We also use LCE-queries on $T$ (\cref{LCE}).

\medskip
\noindent
\textbf{Queries.}
Let us consider a query for $S=T[i \dd j]$ and $A \in \Borders(S)$. If $|A|=1$, we have just one candidate. Otherwise, $A$ is an arithmetic sequence with difference $p$. Let $a = \min(A)$, $A'=A \setminus \{a\}$, and $a'=\min(A')$. We select borders of length $a$ and $a'$ as candidates. If $a' \not \in \allcov(S)$, then \cref{perComb}\eqref{ita} implies that $A \cap \C(S) \subseteq \{a\}$. We also select borders of lengths in $A \cap ((x-2)p,xp]$ as candidates, where $x$ is defined as in \cref{perComb}. Note that there are at most two of them. Let $c$ be the maximum candidate which turned out to be a cover of $S$. Then $A \cap \C(S) = A \cap [1,c]$ by \cref{perComb}\eqref{itb}.

What is left is to compute $x$, that is, the minimum exponent of a run in $S$ with Lyndon root $L$ that is a cyclic shift of $S[1 \dd p]$. Since $|A| \ge 2$, $S$ has a prefix run with Lyndon root $L$. Then $\ell=\min(p+d,|S|)$, where $d=\lcp(i,i+p)$, is the length of the run. If $\ell=|S|$, then $x=\ell/p$ and we are done. Otherwise, let $i'=i+p+d$. We make the following observation.

\begin{claim}
If $a' \in \allcov(S)$, then $T[i' \dd i'+p]$ is contained in a run in $T$ with Lyndon root $L$.
\end{claim}
\begin{proof}
Runs in $T$ with Lyndon root $L$ must cover $S=T[i \dd j]$, since $S[1 \dd a']$ is a periodic cover of $S$ and each of its occurrences is induced by a run. The prefix run in $S$ corresponds to a run ending at position $i'-1$ in $T$. By Observation~\ref{p_1}, the run with Lyndon root $L$ containing the position $i'$ must end after position $i'+p$.
\end{proof}

We identify the run $T[a \dd b]$ with period $p$ containing $T[i' \dd i'+p]$ by asking $\lcp(i',i'+p)$ and $\lcs(i',i'+p)$ queries. This lets us recover the identifier of its Lyndon root $L$. Similarly we compute the suffix run with Lyndon root $L$ in $S$ and the previous run $T[a' \dd b']$ with Lyndon root $L$ in $T$. Using the dictionary on $A_L$, we recover the range in the array that corresponds to elements from $a$ to $a'$. This lets us use a range minimum query on this range in $E_L$ and use it together with the exponents of the prefix and suffix runs of $S$ to compute $x$. All the operations in a query are performed in $\Oh(1)$ time.
\end{proof}

\subsection{Main Query Algorithm}
The main result of this section follows from Lemma~\ref{candidates} and Lemma~\ref{periodic}.
\begin{theorem}
Let $T$ be a string of length $n$. After $\Oh(n \log n)$-time preprocessing, for any factor $S$ of $T$ we can answer a query $\allcov(S)$, with output represented as a union of $\Oh(\log n)$ pairwise disjoint arithmetic sequences, in $\Oh(\log n (\log \log n)^2)$ time.
\end{theorem}

The transformation to the off-line model is similar as in \cref{corMinOffline}.

\begin{corollary}\label{allcoff}
For a string $T$ of length $n$, any $m$ queries $\allcov(T[i\dd j])$ can be answered in $\Oh((n+m) \log n\log \log n)$ time and $\Oh(n+m)$ space.
\end{corollary}
\begin{proof}
As in the proof of \cref{corMinOffline}, we use off-line WA queries and we need to transform the query algorithm to make sure that the $\SeedSets$ tree is processed level by level. The data structure counterpart of \cref{periodic} uses only $\Oh(n)$ space and does not require any transformations. \cref{candidates} for computations on the $\SeedSets$ tree uses only $\IsCover$ queries, which are simpler than $\LCovP$ queries.

Each single such query can be naturally implemented by traversing the tree level by level. Moreover, in each recursive call of the algorithm of \cref{candidates}, all subsequent calls to $\IsCover$ in computing $\chains(B)$ and in the routine $\computeUsing(B,C)$ visit the tree level by level.

Hence, it suffices to simultaneously process the first recursive calls of all $\allcov$ queries, similarly for the second recursive calls etc, where the maximum depth of recursion is $\Oh(\log \log n)$. Thus we will construct the $\SeedSets$ tree, level by level, $\Oh(\log \log n)$ times, each time in $\Oh(n \log n)$ time and $\Oh(n)$ space. 
\end{proof}

\section{Final Remarks}
We showed an efficient data structure for computing internal covers. However, a similar problem for
seeds, which are another well-studied notion in quasiperiodicity, seems to be much harder. We pose the following question.

\medskip\noindent {\bf Open problem.}

\noindent
Can one answer internal queries related to  seeds  in $\Oh(\mbox{polylog\,}n)$ time after $\Oh(n\, \mbox{polylog\,}n)$ time
preprocessing?

\bibliographystyle{plainurl}
\bibliography{icover}

\end{document}